\documentclass[%
 amsmath,
 prl,
 a4paper,
 twocolumn
]{revtex4-2}
\usepackage{hyperref}
\usepackage{graphicx,xcolor}
\usepackage[caption=false]{subfig}
\usepackage{braket,dsfont,amsthm,amssymb}
\newtheorem{theorem}{Observation}

\newtheorem{lemma}[theorem]{Lemma}

\newcommand\norm[1]{\left\lVert#1\right\rVert}

\newcommand{\scq}{\mathcal{S}_{\rm C}}
\DeclareMathOperator{\tr}{Tr}
\newcommand{\id}{\mathds{1}}

\begin{document}
\title{Characterize arbitrary quantum networks in the noisy intermediate-scale quantum era}
\author{Zhen-Peng Xu}
\email{zhen-peng.xu@uni-siegen.de}
\affiliation{School of Physics and Optoelectronics Engineering, Anhui University, 230601 Hefei, China}
\affiliation{Naturwissenschaftlich-Technische
Fakult\"{a}t, Universit\"{a}t Siegen, Walter-Flex-Straße 3, 57068 Siegen, Germany}

\date{\today}


\begin{abstract}
  Quantum networks are of high interest nowadays. In short, it describes the distribution of quantum sources represented by edges to different parties represented by nodes in the network. 
  Bundles of tools have been developed recently to characterize quantum states from the network in the ideal case. However, features of quantum networks in the noisy intermediate-scale quantum (NISQ) era invalidate most of them and call for feasible tools.
  By utilizing purity, covariance and topology of quantum networks, we provide a systematic approach to tackle with arbitrary quantum networks in the NISQ era, which can be noisy, intermediate-scale, random and sparse. 
  One application of our method is to witness the progress of essential elements in quantum networks, like the quality of multipartite entangled sources and quantum memory.
\end{abstract}
\maketitle

Numerous works have advertised from different scales the advent of quantum network technology, as small as the storage of a single entangled pair~\cite{li2021quantum}, and as broad as quantum internet~\cite{kimble2008quantum,sciarrino2012insight,wehner2018quantum}. 
Apart from the theoretical importance, quantum networks appear naturally in practice, especially in quantum key distribution~\cite{gisin2002quantum,yin2020entanglementbased}, quantum network metrology~\cite{komar2014quantum,proctor2018multiparameter,rubio2020quantum} and quantum distributed computation~\cite{caleffi2018quantum}.  
A recent move is into the characterization of different quantum correlations arising from quantum networks~\cite{navascues2020genuine,kraft2021quantum,hansenne2022symmetries,makuta2022no,wang2022quantum,renou2019genuine,pozas2019bounding,gisin2020constraints,aberg2020semidefinite,contreras2021genuine,pozas2022full,tavakoli2021bell,jones2021network}.

A quantum network can be abstracted as a hypergraph, where each node stands for a local lab and each hyperedge represents a quantum source that distributes particles only to labs associated with the corresponding nodes, see Fig.~\ref{fig:inflation} for examples. A correlated quantum network (CQN) allows the pre-shared classical protocol, i.e., global classical correlation~\cite{kraft2021quantum}, an independent quantum network (IQN) allows not.
Despite recent progress~\cite{navascues2020genuine,kraft2021quantum,hansenne2022symmetries,makuta2022no,wang2022quantum,renou2019genuine,pozas2019bounding,gisin2020constraints,aberg2020semidefinite,contreras2021genuine,pozas2022full,tavakoli2021bell,jones2021network}, the study of quantum network states is still in its cradle. 
Past research has focused mainly on IQN, bundles of tools~\cite{pozas2019bounding,kraft2019monogamy,gisin2020constraints,aberg2020semidefinite,kraft2021characterizing} have been added into the current toolbox. 
However, they become incapable to detect the underlying structure of CQN even when only a small amount of global classical correlation appears. 
In comparison, few methods~\cite{navascues2020genuine,kraft2021quantum,navascues2020genuine,hansenne2022symmetries,makuta2022no,wang2022quantum} exist for CQN, which either work only for special kinds of states like symmetric states~\cite{hansenne2022symmetries,wang2022quantum,makuta2022no}, limited quantum networks like the triangle network~\cite{navascues2020genuine,kraft2021quantum} or complete $n$-partite network with $(n-1)$-partite sources~\cite{mao2022test,coiteux2021no,coiteux2021any}. 
However, an undeniable fact is that we progress toward the noisy intermediate-scale quantum (NISQ) era, as pointed out sagaciously by Preskill~\cite{preskill2018quantum}. The global classical correlation exists then frequently in real applications, which can even elicit from the \textit{flap of a butterfly's wings in Brazil}~\cite{ott2008edward}. 

\begin{figure}[htpb]
  \centering
  \includegraphics[width=0.48\textwidth]{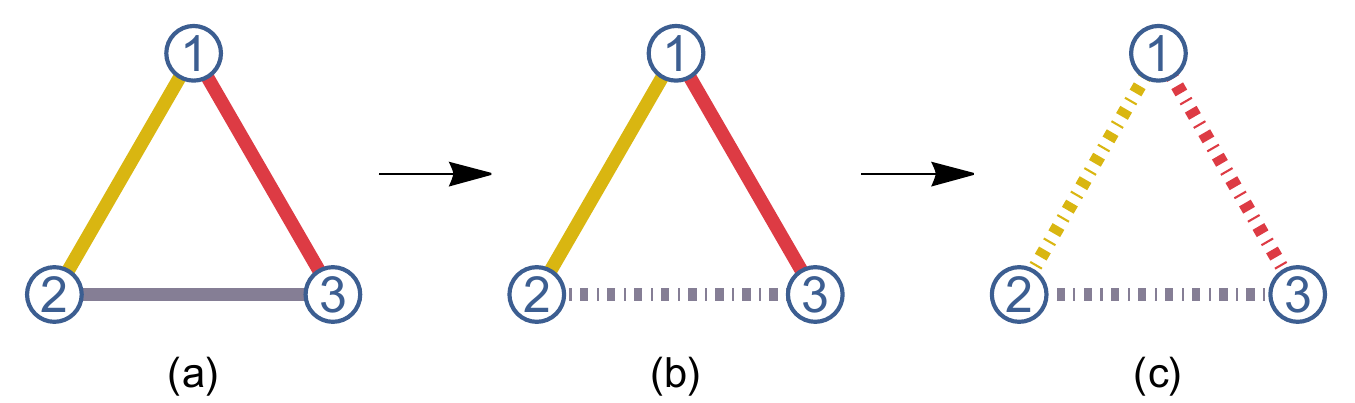}
  \caption{Three quantum networks, where each node stands for one local lab, one edge in real line represents a genuine bipartite entangled quantum source shared by different labs in the corresponding nodes, and one edge in dashed line represents a separable quantum source. In practice, the quantum network in (a) can degenerate to the one in (b), even to the one in (c).} 
  \label{fig:inflation}
\end{figure}

Apart from the unavoidable global noise, quantum networks in the NISQ era share at least other three features: intermediate-scale, random, and sparse. 
Though the size of quantum networks in the NISQ era is limited, it can be not small, considering that IBM has unveiled a quantum chip with $433$ qubits~\cite{choi2023ibm} already. The randomness in the network~\cite{wiersma2010random} can originate from the random establishment of quantum links with quantum repeaters~\cite{briegel1998quantum,sangouard2011quantum}, and also the decoherence of established links as considered in waiting time~\cite{vinay2019statistical}. Degeneration of a triangle quantum network until a classical network is illustrated in Fig.~\ref{fig:inflation}.
Since genuine multipartite entanglement is hard to prepare and to maintain~\cite{mooney2021generation,zhang2022quantum}, the realistic quantum networks will be sparse. Tools for quantum networks in the NISQ era regarding those features are still missing.

In this work, we characterize correlated quantum networks in the NISQ era by employing the purity of the state and covariance of the measured data. Those methods are operational in the sense that only the available experiment data is employed, without knowing the exact underlying quantum state. 
Purity of the network state plays an essential role here, as the classical correlation in a state can be captured by its purity. 
Pretty recent research shows that the purity of a multipartite state can be evaluated efficiently with only local operations~\cite{elben2020mixedstate,huang2020predicting}, which fits the network scenario. 
The methods developed here are feasible for noisy intermediate-scale or big quantum networks. Interestingly, they work even for a collection of networks with different kinds of topology, which can cover the random network models, especially the ones with probabilistic genuine bipartite sources as in the consideration of quantum repeaters~\cite{briegel1998quantum}. Thus, they  answer one corresponding open question in the review paper on nonlocality in quantum networks~\cite{tavakoli2021bell}. 
We can also apply our methods to a part of the network instead of the whole, which fits the sparse structure of the network in the NISQ era and reduces the difficulty of computation.

\textit{GHZ state under decoherence.---} As a warming-up exercise we discuss the Greenberger-Horne-Zeilinger (GHZ) state of $n$ qubits under decoherence, 
\begin{equation}\label{eq:ghz_d}
    \rho(\alpha) = (1-\alpha) |{\rm GHZ}_+\rangle\langle {\rm GHZ}_+| + \alpha |{\rm GHZ}_-\rangle\langle {\rm GHZ}_-|,
\end{equation}
where $|{\rm GHZ}_\pm\rangle = (|0\cdots 0\rangle \pm |1\cdots 1\rangle)/\sqrt{2}$, and $\alpha \in [0,1/2]$ describes the degree of decoherence. 
Despite its simplicity, this example allows us to introduce our main ideas.

If all the $n$ parties implement the same measurement $Z = |0\rangle\langle 0| - |1\rangle \langle 1|$, then two possible combinations of outcomes happen equally with probability $1/2$, i.e., either all of the outcomes are $0$, or all of them are $1$. 
To simulate this statistical behaviour without genuine $n$-partite entanglement, the state for simulation can only be $\rho_c = [|0\cdots 0\rangle\langle 0\cdots 0| + |1\cdots 1\rangle\langle 1\cdots 1|]/2$, since no other $0-1$ string appear as a combination of outcomes. Such a simulation invalidates known methods with only statistical data~\cite{aberg2020semidefinite,mao2022test,coiteux2021no,coiteux2021any}. It costs at least one classical bit of randomness, as the Shannon entropy or the Von Neumann entropy of $\rho_c$ is $1$. However, the Von Neumann entropy of the state $\rho(\alpha)$ is $-[\alpha\log\alpha+(1-\alpha)\log(1-\alpha)]$, which is strictly less than $1$ for $\alpha \in [0, 1/2)$. 
This means that we cannot simulate the statistical behaviour and the Von Neumann entropy of $\rho(\alpha)$ simultaneously by a quantum network with at most $(n-1)$-partite sources.

The Von Neumann entropy is one way to measure the purity of the state, capturing partially the classical correlations in the state. 
To continue, we examine firstly different measures of purity and choose a suitable one for our following methods.
For a given state $\rho$ in the $d$-dimensional space, the common measures of its purity~\cite{horodecki2013quantumness,gour2015resource,streltsov2018maximal} are
R\'{e}nyi $\alpha$-purity $\log_2 d - \log_2(\tr(\rho^\alpha))/(1-\alpha)$, which converges to the Von Neumann entropy as $\alpha$ tends to $1$, and linear entropy purity $\tr(\rho^2) - 1/d$.
Through the whole text, we take $\tau = \tr(\rho^2)$ to quantify the purity, which determines R\'{e}nyi 2-purity and linear entropy purity. The advantage of $\tau$ over other quantifiers, like the Von Neumann entropy, is that it fits the covariance of experimental data well in our approach, as both of them contain the information of $\rho^2$. 
As for the estimation of purity of a multipartite state with different measures, it can be done efficiently with only local operations~\cite{elben2020mixedstate,huang2020predicting}, which are feasible in the network scenario. 

\textit{Noisy quantum networks.---}\label{sec:covariance}
Noise is unavoidable for the quantum network states in the NISQ era, either the local noise or the global one. Quantum networks with different noise models can all be classified as CQNs.
Firstly, we develop the covariance matrix decomposition method for CQN, where a key step is to separate the part related to global classical correlation out in the whole covariance matrix.
For a given hypergraph $G(V,E)$ and a state  $\rho$ from CQN of $G$, the state $\rho$ can be decomposed as
\begin{align}\label{eq:cqn}
    &\rho = \sum\nolimits_k p_k \rho_k, \ \rho_k = \Big(\bigotimes\nolimits_{i\in V}\mathcal{C}_i^{(k)}\Big) \Big( \bigotimes\nolimits_{e\in E} \eta_e^{(k)} \Big),
\end{align}
where $\{p_k\}_k$ with $\sum_k p_k = 1$ and $p_k > 0$ is the global classical correlation, $\mathcal{C}_i^{(k)}$ is a local channel for the $i$-th party, $\eta_e^{(k)}$ is an entangled state distributed from the source labeled by the hyperedge $e$.

For simplicity, we assume each party has only one measurement, and denote $M_i$ the measurement for the $i$-th party. Then we introduce three kinds of covariance matrices, $\Gamma$, $\Gamma^{(k)}$ and  $\Gamma^{(c)}$, whose elements in the $i$-th row and $j$-th column are $\Gamma_{ij}$, $\Gamma_{ij}^{(k)}$ and  $\Gamma^{(c)}_{ij}$, respectively, where  
\begin{align}\label{eq:cov}
    \Gamma_{ij} &= \braket{M_iM_j} - \braket{M_i}\braket{M_j}, \ \braket{M_i} = \tr(\rho M_i), \nonumber\\
    \Gamma^{(k)}_{ij} &= \braket{M_iM_j}_k - \braket{M_i}_k\braket{M_j}_k,\ \braket{M_i}_k = \tr(\rho_k M_i),\nonumber\\
    \Gamma^{(c)}_{ij} &= \sum\nolimits_k p_k \braket{M_i}_k\braket{M_j}_k - \braket{M_i}\braket{M_j}.
\end{align}
The covariance matrix $\Gamma$ is the one that can be observed directly in experiments. The covariance matrices $\{\Gamma^{(k)}\}_k$ are hidden in the experimental data when we assume that the randomness of the sampling $\{p_k, \rho_k\}_k$ is inaccessible. The covariance matrix $\Gamma^{(c)}$ can be viewed as a classical covariance matrix, since it is only about the distribution of classical data $\{\braket{M_1}_k, \ldots, \braket{M_n}_k\}_k$. Throughout the whole paper, we only consider the dichotomic measurements with outcomes $\pm 1$. A pivotal observation is that the classical covariance matrix $\Gamma^{(c)}$ can be separated out from the observed one $\Gamma$ perfectly, i.e.,
\begin{align}\label{eq:dec}
    \Gamma = \sum\nolimits_k p_k \Gamma^{(k)} + \Gamma^{(c)},
\end{align}
whose proof can be found in Sec. A in Supplemental Material (SM)~\cite{sm}.
Since $\{\Gamma^{(k)}\}_k$ are about network states from IQN, the existing method in Ref.~\cite{aberg2020semidefinite} can be employed to impose constraints on them. However, if there is no limitation of $\Gamma^{(c)}$, the observed covariance matrix $\Gamma$ can still have arbitrary relation with the network topology $G$.
As it turns out, the purity of the state implies a nontrivial condition on $\Gamma^{(c)}$,
 leading to a semi-definite programming (SDP) to determine whether a state can arise from CQN with a given topology.
\begin{theorem}\label{thm:decmixed}
For a given state $\rho$ from the CQN with the network topology $G(V,E)$, measurements $\{M_i\}_{i\in V}$, which result in the covariance matrix $\Gamma$, it holds that
\begin{align}\label{eq:decmixed}
    &\Gamma = \sum\nolimits_{e\in E} \Upsilon_e + T,\ \Pi_e \Upsilon_e \Pi_e = \Upsilon_e \succeq 0,\nonumber\\
    &T \succeq 0,\ \max\nolimits_{i\in V} T_{ii} \le \beta,\ \tr(T) \le l_1 \beta,
\end{align}
where $l_1$ is the maximal eigenvalue of  $\sum_{i\in V} M_i\otimes M_i$, $\beta = 2\sqrt{1-\tau^2}$, $T_{ii}$ is the $i$-th diagonal term of $T$, $\Pi_e = \sum_{i\in e} P_i$ with $P_i$ to be the projection onto $i$-th row.
\end{theorem}
To apply the criterion in this observation, we need firstly estimate the purity of the state $\rho$, and then implement the measurements $\{M_i\}_i$ and obtain the covariance matrix from the experimental data.
It should work for arbitrary network topology with around $50$ nodes in practice. 
This observation can be understood as follows.
The term $\sum_{e\in E} \Upsilon_e$ corresponds to $\sum_k p_k \Gamma^{(k)}$, as each $\Gamma^{(k)}$ has a similar decomposition~\cite{aberg2020semidefinite}. The variable $T$ plays the role of $\Gamma^{(c)}$ and inherits all its constraints. A detailed proof is provided in Sec. B in SM~\cite{sm}. 
The application of Observation~\ref{thm:decmixed} to the triangle quantum network is illustrated in Fig.~\ref{fig:triangle_dec}.
We remark that the rank of the state determines the R\'{e}nyi-$0$ purity which reads $\log_2 (d/r)$.
By considering the rank $r$ also, we can set $\beta = \min \{r(1-\tau), 2\sqrt{1-\tau^2}\}$ as a tighter bound.

\begin{figure}
    \centering
    \includegraphics[width=0.33\textwidth]{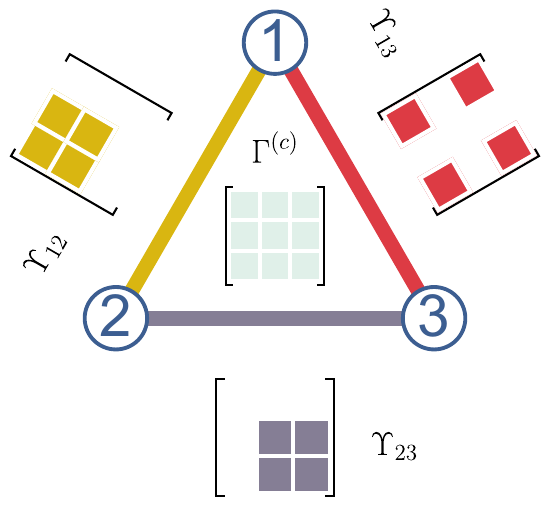}
    \caption{The decomposition of the covariance matrix $\Gamma$ of a noisy state from the triangle network, where each matrix contains $9$ elements, the elements in the blank area are $0$. The block structure of each $\Upsilon_e$ imposes a constraint on itself. The critical step is to obtain constraints of $\Gamma^{(c)}$ from available information of the quantum network, like purity.}
    \label{fig:triangle_dec}
\end{figure}

\textit{Revisit GHZ state under decoherence.---}
We take the state $\rho(\alpha)$ in Eq.~\eqref{eq:ghz_d} and measurements $Z$ for all parties as an example to illustrate Observation~\ref{thm:decmixed}.
The covariance matrix $\Gamma$ of $\rho(\alpha)$ contains always only $1$ as its elements.
  
If $\Gamma$ is from the statistics of a state in a network without $n$-partite sources, then $\Gamma$ should satisfy the decomposition in Eq.~\eqref{eq:decmixed}, where $G$ is the hypergraph with $n$ nodes and includes all subsets with $(n-1)$ elements as hyperedges. Notice that, the rank of $\Gamma$ is $1$, and $\Gamma \succeq \Upsilon_e \succeq 0$, which implies that each $\Upsilon_e$ should be proportional to $\Gamma$. 
  
Since $\Upsilon_e$ always contains element $0$ as exemplified in Fig.~\ref{fig:triangle_dec}, we have $\Upsilon_e = 0$, for all $e\in E$. Consequently, $T = \Gamma$ and $\max_i T_{ii} = 1$. 
The rank of the state $\rho(\alpha)$ is however $2$ and the purity is $\tau = 1-2\alpha+2\alpha^2$. Thus, $\beta = 1$ happens only for $\alpha=1/2$, in which case $\rho(\alpha)$ is fully separable. This leads to the conclusion that $\rho(\alpha)$ can arise from a network without $n$-partite sources if and only if $\alpha = 1/2$. Our criterion is therefore tight. 

\textit{Intermediate-scale networks.---}\label{sec:normineq}
The advantage of covariance matrix decomposition is that it requires only experimental data of few measurements and information of purity. However, the computation becomes heavy for intermediate-scale networks with around $500$ nodes, due to the complexity of SDP in the method.

Here we propose another approach to solve this issue, which can even take care of the randomness feature in the NISQ era.
Firstly, we introduce the fact that $\sum_{ij} |M_{ij}| \le r \tr(M)$ for a semidefinite matrix $M$ whose rank is $r$, and take the triangle network as an example.
According to the decomposition of $\Gamma$ in Observation~\ref{thm:decmixed}, $\sum_{ij} |\Gamma_{ij}| \le \sum_{ij}\sum_e [|\Upsilon_{e,ij}| + |\Gamma^{(c)}_{ij}|] \le \sum_e 2\tr(\Upsilon_{e}) + 3\tr(\Gamma^{(c)})$, where the last inequality is from the block structures of $\Upsilon_{e}$'s and $\Gamma^{(c)}$ as in Fig.~\ref{fig:triangle_dec}. Consequently, $\sum_{ij} |\Gamma_{ij}| \le 2\tr(\Gamma) + \tr(\Gamma^{(c)})$ by applying the first equality in Eq.~\eqref{eq:decmixed} again. 
For the general network topology $G(V,E)$ with $V=\{1, \ldots, n\}$ and $k$ to be the maximal size of hyperedges in $E$, we have 
\begin{equation}\label{eq:anyk}
     \sum\nolimits_{i,j} |\Gamma_{ij}| \le k\tr(\Gamma) + (n-k) \tr(\Gamma^{(c)}).
   \end{equation}
In practice, we can replace $\tr(\Gamma^c)$ in Eq.~\eqref{eq:anyk} by any estimation of it, like the analytical upper bound in Eq.~\eqref{eq:decmixed} results from series of relaxations~\cite{sm}. 
A good estimation plays a vital role in the efficiency of the inequality here, same as in the criterion in Observation~\ref{thm:decmixed}. 
Nowadays, it is still hard to prepare genuine multiparite entangled states for a large system~\cite{mooney2021generation}. Thus, $k$ is usually much smaller than $n$ in Eq.~\eqref{eq:anyk}, i.e., small sources in a big network.

\textit{Random networks.---}
The establishment of genuine multipartite entanglement among remote labs is usually random as in the scenario of quantum repeaters~\cite{briegel1998quantum}. The established one can still degenerate to less-partite ones randomly due to decoherence. This urges us to introduce the concept of random network, where the genuine multipartite entanglement in each source exists probabilistically. 
As an example, we consider a genuine tripartite entangled source, whose degeneration is captured by the triangle network in Fig.~\ref{fig:triangle_dec}, assumed to be with probability $p$. 
The network state $\rho$ has then the decomposition $\rho = p \sum_{i=1}^3 q_i\rho_i + (1-p)\rho_0$, where $\rho_0$ is the original tripartite state and other $\rho_i$'s are independent triangle network states, $\sum_i q_i = 1$ and $q_i \ge 0$. 
This leads to the covariance matrix $\Gamma = p\sum_{i=1}^3 q_i\Gamma^{(i)} + (1-p)\Gamma^{(0)} + \Gamma^{(c)}$, where $\Gamma^{(i)}$'s are the covariance matrices for $\rho_i$'s, and $\Gamma^{(c)}$ is the classical one. As argued before, $\tilde{\Gamma} :=\sum_{i=1}^3 q_i\Gamma^{(i)}$ has the decomposition $\sum_{e\in E} \Upsilon_e$ as in Fig.~\ref{fig:triangle_dec}, which implies that $\sum_{i,j=1}^3 |\tilde{\Gamma}_{ij}| \le 2 \tr(\tilde{\Gamma})$. Consequently, 
\begin{align}
    \sum\nolimits_{i,j} |\Gamma_{ij}| &\le 2 p\tr(\tilde{\Gamma}) + 3[(1-p)\tr(\Gamma^{(0)}) + \tr(\Gamma^{(c)})]\nonumber\\
    &=2\tr(\Gamma) + [(1-p)\tr(\Gamma^{(0)}) + \tr(\Gamma^{(c)})]\nonumber\\
    &\le 2\tr(\Gamma) + 3(1-p) + \tr(\Gamma^{(c)}),
\end{align}
where the last inequality is from the fact that any variance should be no more than $1$ as the outcomes of the measurements are $\pm 1$. 

This result is the very first characterization of random quantum network states, which can be generalized to arbitrary random quantum networks as follows.
\begin{theorem}\label{thm:random}
   Assume $\rho$ is a state from the random quantum network with $n$ parties and $c_k$ genuine $k$-partite sources on average for each $k$. If $\Gamma$ is a covariance matrix of measurements whose outcomes are $\pm 1$, then
  \begin{align}\label{eq:difftopo2}
      \sum\nolimits_{i, j} |\Gamma_{ij}| \le  &\max \sum\nolimits_k k x_k + y\nonumber\\
      \text{such that }&\sum\nolimits_k x_k + y = \tr(\Gamma),\nonumber\\
      &0\le x_k \le k c_k,\nonumber\\
      &0 \le y = \tr(\Gamma^{(c)}).
  \end{align}
\end{theorem}
The proof is in Sec.~C in SM~\cite{sm}. Equation~\eqref{eq:difftopo2} is one inequality including a linear programming, which can be verified efficiently even for large random networks. 

The criterion in Observation~\ref{thm:random} is device-independent, in the sense that it works without any assumption of the underlying quantum system and measurements.
Besides, it does not depend on the exact underlying network topology, but the parameters $\{c_k\}$.
Such results can also be used to benchmark the quality of genuine multipartite entanglement, which degenerates randomly due to decoherence. In such a case, parameters $\{c_k\}$ should be functions of time. Observation~\ref{thm:random} answers an open question in Ref.~\cite{tavakoli2021bell} also, i.e., how to characterize the mixture of quantum networks with different kinds of topology. 

\textit{Sparse networks.---} In a reasonable prospect, the quantum network should be sparse in the near future. Even though we have a relatively large quantum network, the size and the amount of quantum sources would be relatively small as illustrated in Fig.~\ref{fig:sparse}. The exact numbers depend on the progress of quantum technologies.

\begin{figure}[h]
    \centering
    \includegraphics[width=0.315\textwidth]{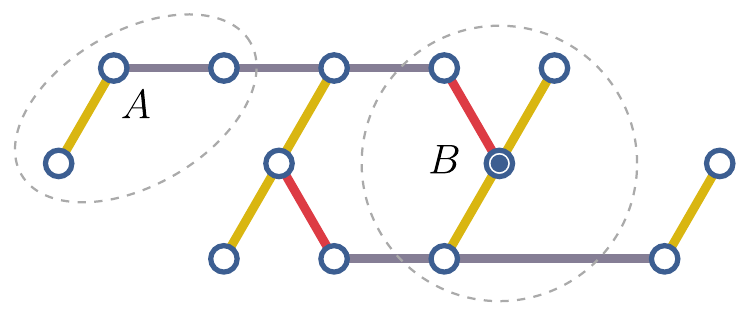}
    \caption{A sparse network $G_s$ with $13$ nodes, which contains an angle (in region $A$) as a sub-network. The centeral node in region $B$ is the filled one.}
    \label{fig:sparse}
\end{figure}

For a given state $\rho$ from a large network $G(V,E)$, a necessary condition is that $\rho_S := \tr_{\bar{S}} \rho$ can arise from the network with the induced subgraph $H(S,E_S)$ on $S$, where $S$ is an arbitrary subset of $V$ and $E_S = \{e\cap S|e\in E\}$.
Then we could apply Observations~\ref{thm:decmixed} and \ref{thm:random} for each reduced state $\rho_S$. The sparsity of the network can reduce the complexity of this approach, as we do not need too many relatively small sub-networks to cover the original one. 
Correspondingly, we have to estimate the classical covariance matrix $\Gamma^{(c)}_S$ for each $\rho_S$. Since $\rho_S$ as a reduced state could be much mixed than the original state $\rho$, the purity of $\rho_S$ could lead to too loose constraints in Observation~\ref{thm:decmixed}. A key observation here is that $\Gamma^{(c)}_{S,ij} = \Gamma^{(c)}_{ij}$ for any $i,j \in S$, with $\Gamma^{(c)}$ to be the classical covariance matrix corresponding to the original network state $\rho$. More explanations are provided in Sec. D in SM~\cite{sm}.

For instance, we consider the state $\rho(\alpha)$ as in Eq.~\eqref{eq:ghz_d} for $13$ qubits and network $G_s$ with only bipartite sources as shown in Fig.~\ref{fig:sparse}. Then the reduced state $\rho_A(\alpha) = (|000\rangle\langle 000| + |111\rangle\langle 111|)/2$ for the three qubits in region $A$ as in Fig.~\ref{fig:sparse}. The rank of $\rho_A(\alpha)$ is two, and the covariance matrix of $\rho_A(\alpha)$ contains only $1$. 
Notice that, the sub-network in the region $A$ is a special case of the triangle network as illustrated in Fig.~\ref{fig:inflation}. The same argument as before implies that $\Gamma^{(c)}_{A}$ contains $1$ only, which does not contradict with the purity of $\rho_A(\alpha)$, but with the one of the original state $\rho(\alpha)$ for $\alpha \in [0,1/2)$. Thus, with the global purity and the statistical data of the qubits in the small region $A$, we obtain the same tight result. This strategy saves much effort, since we only need to measure few qubits in a large network.

There is another approach to employ sub-networks to determine whether a state can arise from the original network or not, i.e., we measure out one party associated with one node $v$ and broadcast the outcomes. Then we can treat the party associated with node $v$ and all the related sources together as a new multipartite source, which distributes particles to all the parties in ${\cal N}(v) := \{u| (u,v) \in E\}$. 
If the original network is sparse, then the size of ${\cal N}(v)$, i.e., the size of the new introduced entangled source is usually not big.
For example, if $v$ is the central node in region $B$ as in Fig.~\ref{fig:sparse}, then the new source distributes particles to parties associated with all the other three nodes in the region $B$.
We can do this procedure for a subset $S$ of nodes sequentially.
By applying Observations~\ref{thm:decmixed} and \ref{thm:random} for the resulting sub-network, we can obtain new criteria for the original network state.

\textit{Discussion.---}\label{sec:conclusion}
Quantum networks work as a playground for various quantum technologies, like quantum repeaters and quantum memory. Concerning the real-life implementation of quantum networks, we examine them in the background of the noisy intermediate-scale quantum (NISQ) era. In this paper, we have focused on four aspects of such quantum networks, that is, they should be noisy, intermediate-scale, random, and sparse. We developed operational methods based on purity and covariance to address all those four features.

There exist already methods to tackle with the noisy quantum network states, e.g., the witness based on fidelity~\cite{navascues2020genuine,kraft2021quantum,makuta2022no,wang2022quantum} and nonlocality inequalities~\cite{coiteux2021any,coiteux2021no,mao2022test}. However, the witness based on fidelity works  mostly either for small networks or special states like graph states in practice. The nonlocality inequalities are designed specially for the $n$-partite network with all $(n-1)$-partite sources. In comparison, our methods work for any kind of network topology, by employing  experiment data only, without knowing the exact underlying quantum state and measurements. Nevertheless, quantum theory is assumed here, which is another difference between our consideration and network nonlocality.

`Quantum technologists should continue to strive for more accurate quantum gates and, eventually, fully fault-tolerant quantum computing'~\cite{preskill2018quantum} and networks, for which our methods can provide the witness.

\textit{Acknowledgements.---}
I thank David Gross, H Chau Nguyen, Julio I. de Vicente, Kiara Hansenne, Mariami Gachechiladze, Nikolai Wyderka, Otfried G\"{u}hne, Shu-Heng Liu, Sixia Yu, Tristan Kraft, especially Laurens Ligthart and Thomas Cope for discussions and suggestions. Moreover, I would like to thank Tristan Kraft for careful reading of the manuscript. This work was supported by {National Natural Science Foundation of China} (Grant No.\ 12305007) and 
Anhui Provincial Natural Science Foundation (Grant No.\ 2308085QA29), the Deutsche Forschungsgemeinschaft (DFG, German Research Foundation, project numbers 447948357 and 440958198), the Sino-German Center for Research Promotion (Project M-0294), the ERC (Consolidator Grant 683107/TempoQ),  the Alexander Humboldt foundation.

\onecolumngrid
\newpage

\begin{center}
    \textbf{Supplemental Material of\\ ``Characterize arbitrary quantum networks in the noisy intermediate-scale quantum era"}\\ \vspace{0.3em}
    Zhen-Peng Xu\\
    \textit{School of Physics and Optoelectronics Engineering, Anhui University, 230601 Hefei, China and}\\
    \textit{Naturwissenschaftlich-Technische Fakult\"{a}t, Universit\"{a}t Siegen, Walter-Flex-Stra\ss e 3, 57068 Siegen, Germany}
\end{center}

\setcounter{theorem}{4}
\section{A. Covariance Matrix of mixed state}\label{sec:covariance}
For a given state $\rho = \sum_k p_k \rho_k$ and a set of observables $\{M_i\}$, denote 
\begin{align}
    &\Gamma_{ij} = \tr(M_iM_j\rho) - \tr(M_i\rho)\tr(M_j\rho),\\ 
    &\Gamma_{ij}^{(k)} = \tr(M_iM_j\rho_k) - \tr(M_i\rho_k)\tr(M_j\rho_k),\\ 
    &\Gamma_{ij}^c = \sum_k p_k \tr(M_i\rho_k)\tr(M_j\rho_k) - \tr(M_i\rho)\tr(M_j\rho).
\end{align}
Then we have
\begin{equation}
    \Gamma_{ij} = \sum_k p_k \Gamma_{ij}^{(k)} + \Gamma_{ij}^c.
\end{equation}
Besides, the matrices $\Gamma = (\Gamma_{ij})$,  $\Gamma^{(k)} = (\Gamma_{ij}^{(k)})$,  $\Gamma^c = (\Gamma_{ij}^c)$ are positive semi-definite. 
\section{B. Estimation of mean value and variance}

\begin{lemma}
  If the mixed state $\rho$ has the decomposition
   \begin{equation}
    \rho = \sum_k p_k \rho_k, p_k > 0,
  \end{equation}
  then
   \begin{equation}
     \operatorname{rank}(\rho-\rho_k) \le \operatorname{rank}(\rho).
  \end{equation}
\end{lemma}
\begin{proof}
  Denote $\Pi$ the projection onto the subspace $S_{\rho}$ spanned by the eigenstates of $\rho$.  
  By definition,
  \begin{equation}
    \Pi \rho \Pi = \rho,
  \end{equation}
  which implies that
  \begin{equation}
    \sum_k p_k (\rho_k - \Pi\rho_k\Pi) = 0. 
  \end{equation}
  Since $\rho_k - \Pi\rho_k\Pi \succeq 0$, we have $\Pi\rho_k\Pi = \rho_k$, $\forall k$.
  Hence, $\operatorname{rank}(\rho-\rho_k)$ is no more than the dimension of $S_{\rho}$, which, by definition, is  $\operatorname{rank}(\rho)$.
\end{proof}

\begin{lemma}\label{le:traced}
  For any projector $P$, and two states $\rho, \sigma$ where $\sigma$ is in the range of $\rho$, denote $M = 2P - \id$, we have
   \begin{equation}
     \tr(M(\rho-\sigma)) \le \sqrt{R} \norm{\rho-\sigma}_{\rm F}  ,
  \end{equation}
  where $R = \max_{n\le \operatorname{rank}(\sigma)} \frac{4n(r-n)}{r}$, $r = \operatorname{rank}(\rho)$ and $ \norm{\cdot}_{\rm F}$ is the Frobenius norm.
\end{lemma}
\begin{proof}
  \begin{align}\label{eq:singular}
    \tr(P(\rho-\sigma)) = \tr(P(\rho-\sigma)^{+}) - \tr(P(\rho-\sigma)^{-})
                        \le \tr(P(\rho-\sigma)^{+})
                        \le \tr((\rho-\sigma)^{+})
                        = \sum_{\lambda_i>0} \lambda_i,
\end{align}
where $(\rho-\sigma)^{\pm}$ are the non-negative and negative part of  $\rho-\sigma$, that is, 
\begin{equation}
   \rho-\sigma = (\rho-\sigma)^{+} - (\rho-\sigma)^{-}.
\end{equation}
$\tr(\rho - \sigma) = 0$ implies that
\begin{equation}
  \sum_{\lambda_i>0} \lambda_i + \sum_{\lambda_i<0} \lambda_i = 0.
\end{equation}
Denote $p, n$ the number of positive eigenvalues and the number of negative eigenvalues of  $\rho-\sigma$, respectively.
\begin{align}\label{eq:frobenius}
   \tr((\rho-\sigma)^2) &= \sum_{\lambda_i>0} \lambda_i^2 + \sum_{\lambda_i<0} \lambda_i^2 \ge \frac{1}{p} \left(\sum_{\lambda_i>0}\lambda_i\right)^2 + \frac{1}{n} \left(\sum_{\lambda_i<0}\lambda_i\right)^2 = \frac{n+p}{np} \left(\sum_{\lambda_i>0}\lambda_i\right)^2.
\end{align}
Combining Eq.~\eqref{eq:singular} and Eq.~\eqref{eq:frobenius}, we have 
\begin{align}
    \tr(P(\rho-\sigma)) \le \sqrt{\frac{np}{n+p}} ||\rho-\sigma||_{\rm F} \le \sqrt{\frac{n(r-n)}{r}} ||\rho-\sigma||_{\rm F},
\end{align}
since $n+p \le r$.
Notice that $\tr(M(\rho-\sigma)) = 2\tr(P (\rho-\sigma))$ and $n \le \operatorname{rank}(\sigma)$ as observed in Ref.~\cite{coles2019strong}, we finish the proof.
\end{proof}

\begin{theorem}\label{thm:puritybound}
  For a given set of dichotomic observables $\{M_i\} $ with outcome $\pm 1$ and a mixed state $\rho = \sum_k p_k \rho_k$, we have
  \begin{equation}
    \left|\sum_k p_k \braket{M_i}_k \braket{M_j}_k - \braket{M_i}\braket{M_j} \right| \le R(\tau_0-\tau),
  \end{equation}
  where $R = \max_{n\le n_0} \frac{4n(r-n)}{r}$, $r = \operatorname{rank}(\rho)$, $n_0 = \max_k \operatorname{rank}(\rho_k) \le r$, $\tau_0$ is the average purity, i.e,
 $\tau_0 = \sum_k p_k \tr(\rho_k^2)$.
\end{theorem}
\begin{proof}
\begin{align*}
  \left|\sum_k p_k \braket{M_i}_k \braket{M_j}_k - \braket{M_i}\braket{M_j} \right| &=   \left|\sum_k p_k (\braket{M_i}_k - \braket{M_i})(\braket{M_j}_k-\braket{M_j})\right|\\
                                                                                   &\le \sum_k p_k  |{\braket{M_i}_k-\braket{M_i}}|  |{\braket{M_j}_k-\braket{M_j}}|\\
                                                                                   &\le \sum_k p_k \left( \sqrt{R}  \norm{\rho-\rho_k}_{\rm F}   \right)^2\\
                                                                                   &= R\sum_k p_k \tr((\rho-\rho_k)^2)\\
                                                                                   &= R \sum_k p_k (\tr(\rho_k^2)+\tr(\rho^2)-2\tr(\rho\rho_k))\\
                                                                                   &= R(\tau_0-\tau).
\end{align*}
\end{proof}
Note that $R \le r$ and $\tau_0 \le 1$.
In the case that $\rho_k$'s are all pure states,  $\tau_0 = 1$ and $R \le \min\{4(1-1/r), r\}$.

{\color{black}
We have two remarks. Firstly, here we have made use of the rank $r$ and the quantifier of purity $\tau = \tr(\rho^2)$ to provide an upper bound. In principle, there could be other quantifiers of purity which can be employed in a tighter bound. The crucial point is how to get rid of the exact decomposition in the procedure of relaxation, since only the final state $\rho$ is assumed to be available.
Secondly, another measure of purity is the single-shot distillable purity $\mathcal{P}_d^1(\rho)$, which equals to $\lfloor \log_2 (d/r) \rfloor$~\cite{streltsov2018maximal}. Since the rank $r$ determines $\mathcal{P}_d^1(\rho)$, but not the other way around. If we employ $\mathcal{P}_d^1(\rho)$ instead of the rank $r$, the results might be less accurate in principle. Hence, it is also a key point that how to chose and combine different measures of purity to extract more information.
}

\begin{theorem}
  For a set of hermitian operators $\{M_i\} $, a mixed state $\rho = \sum_k p_k \rho_k$,
  \begin{equation}\label{eq:tighttrace}
    \tr(\Gamma_p) \le l_1 r (\tau_0-\tau),
  \end{equation}
  where $l_1$ is the maximal singular value of  $\sum_i M_i\otimes M_i$, $r,\tau, \tau_0$ are the rank and purity of  $\rho$ and the average purity of the decomposition, respectively.
\end{theorem}
\begin{proof}
  \begin{align}
    \tr(\Gamma_p) &= \sum_i \left(\sum_k p_k \braket{M_i}_k\braket{M_i}_k - \braket{M_i}\braket{M_i}\right)\\
                  &= \sum_i \left(\sum_k p_k \left(\braket{M_i}_k - \braket{M_i}\right)^2\right)\\
                  &= \tr\left(\left(\sum_i M_i\otimes M_i\right) \left(\sum_k p_k (\rho_k - \rho)\otimes (\rho_k - \rho) \right)\right)\\ 
                  &\le l_1  \norm{\sum_k p_k (\rho_k - \rho)\otimes(\rho_k-\rho)}_1\\
                  &= 2l_1 \sum_{\lambda_i>0} \lambda_i,
  \end{align}
  where $\{\lambda_i\}$ are eigenvalues of  $\sum_k p_k (\rho_k - \rho)\otimes(\rho_k - \rho)$.

  Note that, 
  \begin{equation}
    \operatorname{rank}\left(\sum_k p_k (\rho_k- \rho)\otimes(\rho_k-\rho)\right) \le r^2.
  \end{equation}
  Following the similar procedure in Lemma~\ref{le:traced}, we know that
  \begin{equation}
    \sum_{\lambda_i>0} \lambda_i \le \frac{\sqrt{r^2}}{2} \norm{\sum_k p_k (\rho_k - \rho)\otimes(\rho_k-\rho)}_F,
  \end{equation}
  which leads to
  \begin{align}
    \tr(\Gamma_p)  &\le l_1r \norm{\sum_k p_k (\rho_k - \rho)\otimes(\rho_k-\rho)}_F \\ 
                   &= l_1 r \left(\sum_{k,t} p_kp_t [\tr ((\rho_k-\rho)(\rho_t-\rho))]^2\right)^{1/2}\\
                   &\le l_1 r \left(\sum_{k,t} p_kp_t \tr((\rho_k-\rho)^2)\tr((\rho_t-\rho)^2)\right)^{1/2}\\
                   &= l_1 r \sum_k p_k \tr((\rho_k-\rho)^2)\\
                   &= l_1 r(\tau_0-\tau).
  \end{align}
\end{proof}
In the case that $\{M_i\} $ are all dichotomic measurements, $l_1$ is no more than the size of  $\{M_i\} $. Hence, the bound in Eq.~\eqref{eq:tighttrace} is tighter than the one by trivially summing up the upper bound $r(\tau_0-\tau)$ for each term in the diagonal.

In summary, we have proved the following observation. Notice that $\tau_0 \le 1$.
\begin{theorem}\label{thm:classicalcov}
  For a given set of dichotomic measurements ${\cal M}$, a state $\rho$ with rank $r$ and purity  $\tau$,
  \begin{equation}\label{eq:classicalcov}
    \max \Gamma^c \le r(1-\tau),\ \tr(\Gamma^c) \le l_1 r(1-\tau),
  \end{equation}
  where $l_1$ is the maximal singular value of  $\sum_{M\in {\cal M}} M\otimes M$.
\end{theorem}

\begin{theorem}
  \begin{align*}
    \max \Gamma^c \le 2\sqrt{1-\tau^2},\ \tr(\Gamma^c) \le 2l_1 \sqrt{1-\tau^2},
  \end{align*}
    where $l_1$ is the maximal singular value of  $\sum_i M_i\otimes M_i$.
\end{theorem}
\begin{proof}
  As we have observed,
  \begin{align}
    &\Gamma^c_{ij} = \braket{M_i}\braket{M_j} - \sum_k p_k \braket{M_i}_k \braket{M_j}_k = \tr\left([M_i\otimes M_j]\left[\rho\otimes\rho - \sum_k p_k \rho_k\otimes\rho_k\right]\right),\\
    &\tr(\Gamma^c) = \tr\left(\left[\sum_i M_i\otimes M_i\right]\left[\rho\otimes\rho - \sum_k p_k \rho_k\otimes\rho_k\right]\right).
  \end{align}
  Since the maximal singular values of $M_i\otimes M_j$ is $1$, Von Neumann's trace inequality, and the relation between trace norm and fidelity imply that, $\forall \epsilon\in [0,1)$,
  \begin{align}
    \left|\Gamma^c_{ij}\right| &\le  \norm{\rho\otimes\rho - \sum_k p_k \rho_k\otimes\rho_k}_{\tr} \le 2\sqrt{1-f\left(\rho\otimes\rho, \sum_k p_k \rho_k\otimes\rho_k\right)} \le 2\sqrt{1-\tau^2},
  \end{align}
  where $\epsilon\in [1/3,\tau)$ and the last inequality is from the fact that
  \begin{align}
    f(\rho\otimes\rho, \sum_k p_k \rho_k\otimes\rho_k) &= \left[ \tr\sqrt{ \sqrt{\rho\otimes\rho} \left(\sum_k p_k \rho_k\otimes \rho_k\right) \sqrt{\rho\otimes\rho}} \right]^2  \\
    &\ge \tr\left( \sqrt{\rho\otimes\rho} \left(\sum_k p_k \rho_k\otimes \rho_k\right) \sqrt{\rho\otimes\rho} \right)\\
    &\ge \tr\left( \left(\rho\otimes\rho\right) \left(\sum_k p_k \rho_k\otimes \rho_k\right) \right)\\
    &\ge \left(\sum_k p_k \tr(\rho\rho_k)\right)^2 = \tau^2.
  \end{align}
  Similarly, we can proof the result for $\tr(\Gamma^c)$.
\end{proof}

We have two extra remarks. 
Firstly, for a given mixed state $\rho = \sum_k p_k \rho_k$ whose purity is $\tau$, we have
\begin{align}
    1-\tau = 1-\tr(\rho^2) \ge \sum_k p_k \norm{\rho-\rho_k}_{\rm F}^2.
\end{align}
Hence, there is a $k$ such that $\norm{\rho-\rho_k}_{\rm F} \le \sqrt{1-\tau}$.

Secondly, Von Neumann's trace inequality and the relation between trace distance and fidelity imply that
\begin{equation}
  \tr(M(\rho-\sigma)) \le \norm{\rho-\sigma}_{\tr} \le 2\sqrt{1 - f(\rho,\sigma)}.
\end{equation}
Besides, $\tau = \tr(\rho^2) \le \sum_k p_k f(\rho,\rho_k)$ implies that there exist a $k'$ such that $f(\rho,\rho_{k'}) \ge \tau$. Consequently, $|\tr(M(\rho-\rho_{k'}))| \le 2\sqrt{(1-\tau)}$.

\section{C. Covariance inequalities}

For a given covariance matrix $\Upsilon$ whose rows and columns are divided into $k$ blocks, we have
\begin{align}\label{eq:basicnormineq}
  \sum_{i,j=1}^k  |{\Upsilon_{ij}}| &\le \sum_{i,j=1}^k \sqrt{ |{\Upsilon_{ii}}| |{\Upsilon_{jj}}| } = \left(\sum_{i=1}^k \sqrt{ |{\Upsilon_{ii}}|  }\right)^2 \le k \sum_{i=1}^k  |{\Upsilon_{ii}}| = k\tr(\Upsilon).
\end{align}
For a network $G(V,E)$, where $V$ is the set of nodes or receivers, $E = \{e\}$ is the set of sources, denote 
$  C_{ij} = \{e \mid e \supseteq (i,j)\}$, $c_{ij}$ the size of $C_{ij}$.
For a given state $\rho \in \scq$, from the covariance matrix decomposition, we know that
$  \sum_{e\in C_{ij}} \Upsilon_{e,ij} = \Gamma_{ij} - \Gamma^c_{ij}$.
\begin{align}\label{eq:covdeca}
  \sum_{i,j}  |{\Gamma_{ij}}| &\le \sum_{i,j} \left(\sum_e  |{\Upsilon_{e,ij}}| +  |{\Gamma^c_{ij}}|\right) \le \sum_e k_e \tr(\Upsilon_e) + n\tr(\Gamma^c),
\end{align}
where $k_e$ is the size of the source $e$.
In the last inequality, we have applied the inequality in Eq.~\eqref{eq:basicnormineq} for both of $\Upsilon_{e,ij}$ and $\Gamma^c_{ij}$.

Denote $k=\max_e k_e$, we have
\begin{align}\label{eq:normineq1}
   \sum_{i,j}  |{\Gamma_{ij}}| &\le k\sum_e \tr(\Upsilon_e) + n\tr(\Gamma^c) = k\tr(\Gamma) + (n-k)\tr(\Gamma^c).
\end{align}
This leads to Corollary 1 in the main text.

In the case of an IQN with only bipartite sources and at most $c_3$ tripartite sources, $\Gamma^c = 0$. Similarly, we have
 \begin{equation}
   \sum_{i,j}  |{\Gamma_{ij}}| \le 2\tr(\Gamma) + \sum_{e, k_e=3} \tr(\Upsilon_e) \le 2\tr(\Gamma) + 3m c_3, 
\end{equation}
where $m$ is the maximal number of measurements per party. Here we have made use of the inequality  $\tr(\Upsilon_e) \le k_e m$, since the covariance of the random variable in  $[-1,1]$ is no more than  $1$.

For any state $\rho=\sum_k p_k\rho_k$, where  $\rho_k$ is from an IQN with only bipartite sources and at most  $c_3$ tripartite sources, we have
\begin{align}
   \sum_{i,j}  |{\Gamma_{ij}}| & = \sum_{i,j}  \Big|{\sum_k p_k \Gamma^{(k)}_{ij} + \Gamma^c_{ij}}\Big| \\
                                        &\le \sum_k p_k \sum_{i,j}  |{\Gamma^{(k)}_{ij}}|  + \sum_{i,j} |{\Gamma^c_{ij}}|\\
                                  &\le \sum_k p_k(2\tr(\Gamma^{(k)}) + 3mc_3) + n\tr(\Gamma^c)\\
                                  &= 2\tr(\Gamma) + (n-2)\tr(\Gamma^c) + 3mc_3.
\end{align}
{Note that, $\rho_k$'s might be from different networks with different topologies, the only constraint is that there are only bipartite sources and at most  $c_3$ tripartite sources.}

In the general case, 
\begin{equation}\label{eq:generalIQN}
\sum_{i,j}  |{\Gamma^{(k)}_{ij}}| \le \sum_e k_e \tr(\Upsilon_e^{(k)}),
\end{equation}
which leads to
\begin{align}
    \sum_{i,j}  |{\Gamma_{ij}}| & \le \sum_k p_k \sum_e k_e \tr(\Upsilon_e^{(k)}) + n \tr(\Gamma^c_{ij})\nonumber\\
    & = \sum_t t \sum_k p_k \sum_{e\in E_t^{(k)}} \tr(\Upsilon_e^{(k)}) + n \tr(\Gamma^c_{ij})\nonumber\\
    &= \sum_t t x_t + n y,
\end{align}
where $x_t:= \sum_k p_k \sum_{e\in E_t^{(k)}} \tr(\Upsilon_e^{(k)})$, and $y := \tr(\Gamma^c_{ij})$. 

By definition, we have $\sum_t x_t + y = \tr(\Gamma)$.
Notice that
\begin{align}
     0\le x_t &\le  \sum_k p_k \sum_{e\in E_t^{(k)}} m t =  \sum_k p_k c_t^{(k)} m t = c_t m t,
\end{align}
with $c_t = \sum_k p_k c_t^{(k)}$ to be the average number of genuine $k$-partite sources.

\section{D. Sub-networks}
For a given network $G(V,E)$, a network state $\rho$, and a subset $S$ of $V$, denote $\rho_S$ the reduced state of $\rho$ on $S$.
By definition, the state $\rho$ can be decomposed as
\begin{align}\label{eq:cqn}
    &\rho = \sum\nolimits_k p_k \rho_k, \ \rho_k = \Big(\bigotimes\nolimits_{i\in V}\mathcal{C}_i^{(k)}\Big) \Big( \bigotimes\nolimits_{e\in E} \eta_e^{(k)} \Big),
\end{align}
where $\{p_k\}_k$ with $\sum_k p_k = 1$ and $p_k > 0$ is the global classical correlation, $\mathcal{C}_i^{(k)}$ is a local channel for the $i$-th party, $\eta_e^{(k)}$ is an entangled state distributed from the source labeled by the hyperedge $e$.

Firstly, the decomposition $\rho = \sum_k p_k \rho_k$ leads to the decomposition $\rho_S = \sum_k p_k \rho_{S,k}$ with $\rho_{S,k}$ to be the corresponding reduced state of $\rho_k$, and $\rho_k$ is an independent network state of the original network implies that $\rho_{S,k}$ is also an independent network state of the sub-network, by definition of the state from correlated quantum networks. Secondly, $\tr(M\rho_{S,k}) = \tr(M\rho_k)$ if $M$ acts only nontrivially on the subset $S$. Then the definition of the classical covariance matrices implies that $\Gamma^{(c)}_{S,ij} = \Gamma^{(c)}_{ij}$, and $\Gamma^{(c)}_{S}$ inherits all the constraints for $\Gamma^{(c)}$.
To be more explicitly,
\begin{align}
    \Gamma^{(c)}_{S,ij} &= \sum_k p_k  \langle M_i\rangle_{S,k}\langle M_j\rangle_{S,k} - \langle M_i\rangle_{S}\langle M_i\rangle_{S} = \sum_k p_k  \langle M_i\rangle_{k}\langle M_j\rangle_{k} - \langle M_i\rangle\langle M_i\rangle = \Gamma^{(c)}_{ij}.
\end{align}

\bibliographystyle{apsrev4-2}
\bibliography{network_states.bib}

\end{document}